\documentclass[11pt]{article}

\usepackage{fullpage}

\usepackage{times}
\usepackage{helvet}
\usepackage{courier}
\usepackage{amssymb,amsmath}
\usepackage{comment}
\usepackage{graphicx}

\usepackage{stmaryrd}

\usepackage{amsmath}
\usepackage{comment}
\usepackage{amssymb}

\newcommand{\ignore}[1]{}

\begin{document}

\title{Homomorphisms are indeed a good basis for counting:\\Three fixed-template dichotomy theorems, for the price of one}

\author{Hubie Chen\\
Departamento LSI\\
Facultad de Inform\'{a}tica\\
Universidad del Pa\'{i}s Vasco\\
San Sebasti\'{a}n, Spain\\
\emph{and}\\
 IKERBASQUE, Basque Foundation for Science\\
Bilbao, Spain\\}

\date{ } 

\maketitle

\newcommand{\confversion}[1]{}
\newcommand{\longversion}[1]{#1}


\newtheorem{theorem}{Theorem}[section]
\newtheorem{conjecture}[theorem]{Conjecture}
\newtheorem{corollary}[theorem]{Corollary}
\newtheorem{proposition}[theorem]{Proposition}
\newtheorem{prop}[theorem]{Proposition}
\newtheorem{lemma}[theorem]{Lemma}
\newtheorem{remarkcore}[theorem]{Remark}
\newtheorem{exercisecore}[theorem]{Exercise}
\newtheorem{examplecore}[theorem]{Example}
\newtheorem{assumptioncore}[theorem]{Assumption}

\newenvironment{assumption}
  {\begin{assumptioncore}\rm}
  {\hfill $\Box$\end{assumptioncore}}

\newenvironment{example}
  {\begin{examplecore}\rm}
  {\hfill $\Box$\end{examplecore}}

\newenvironment{exercise}
 {\begin{exercisecore}\rm}
 {\hfill $\Box$\end{exercisecore}}

\newenvironment{remark}
 {\begin{remarkcore}\rm}
 {\hfill $\Box$\end{remarkcore}}

\newenvironment{proof}{\noindent\textbf{Proof\/}.}{$\Box$ \vspace{1mm}}

\newtheorem{researchq}{Research Question}

\newtheorem{newremarkcore}[theorem]{Remark}

\newenvironment{newremark}
  {\begin{newremarkcore}\rm}
  {\end{newremarkcore}}

\newtheorem{definitioncore}[theorem]{Definition}

\newenvironment{definition}
  {\begin{definitioncore}\rm}
  {\end{definitioncore}}

\newcommand{\Q}{\mathbb{Q}}

\newcommand{\ppequiv}{\mathsf{PPEQ}}
\newcommand{\eq}{\mathsf{EQ}}
\newcommand{\iso}{\mathsf{ISO}}
\newcommand{\ppeq}{\ppequiv}
\newcommand{\ppiso}{\mathsf{PPISO}}
\newcommand{\boolppiso}{\mathsf{BOOL}\mbox{-}\mathsf{PPISO}}
\newcommand{\csp}{\mathsf{CSP}}
\newcommand{\scsp}{\mathsf{SCSP}}
\newcommand{\gi}{\mathsf{GI}}
\newcommand{\ci}{\mathsf{CI}}
\newcommand{\rela}{\mathbf{A}}
\newcommand{\relb}{\mathbf{B}}
\newcommand{\relc}{\mathbf{C}}
\newcommand{\alga}{\mathbb{A}}
\newcommand{\algb}{\mathbb{B}}
\newcommand{\algab}{\mathbb{A}_{\relb}}

\newcommand{\idemp}{I}

\newcommand{\varv}{\mathcal{V}}
\newcommand{\variety}{\mathcal{V}}
\newcommand{\false}{\mathsf{false}}
\newcommand{\true}{\mathsf{true}}
\newcommand{\pol}{\mathsf{Pol}}
\newcommand{\inv}{\mathsf{Inv}}
\newcommand{\cc}{\mathcal{C}}
\newcommand{\alg}{\mathsf{Alg}}
\newcommand{\pitwo}{\Pi_2^p}
\newcommand{\sigmatwo}{\Sigma_2^p}
\newcommand{\pithree}{\Pi_3^p}
\newcommand{\sigmathree}{\Sigma_3^p}

\newcommand{\fancyg}{\mathcal{G}}
\newcommand{\tw}{\mathsf{tw}}

\newcommand{\mc}{\mathsf{MC}}
\newcommand{\mcs}{\mathsf{MC}_s}

\newcommand{\mcb}{\mathsf{MC_b}}

\newcommand{\qc}{\mathrm{QC}}
\newcommand{\normqc}{\mathrm{norm\mbox{-}QC}}
\newcommand{\rqc}{\mathsf{RQC\mbox{-}MC}}

\newcommand{\qcfo}{\mathrm{QCFO}}
\newcommand{\qcfofk}{\qcfo_{\forall}^k}
\newcommand{\qcfoek}{\qcfo_{\exists}^k}

\newcommand{\fo}{\mathrm{FO}}
\newcommand{\fok}{\mathrm{FO}^k}

\newcommand{\tup}[1]{\overline{#1}}

\newcommand{\nn}{\mathsf{nn}}
\newcommand{\bush}{\mathsf{bush}}
\newcommand{\width}{\mathsf{width}}

\newcommand{\un}{N^{\forall}}
\newcommand{\en}{N^{\exists}}

\newcommand{\ord}{\tup{u}}
\newcommand{\ordp}[1]{\tup{#1}}

\newcommand{\gc}{G^{-C}}

\newcommand{\ar}{\mathrm{ar}}
\newcommand{\free}{\mathsf{free}}
\newcommand{\vars}{\mathsf{vars}}

\newcommand{\qed}{}

\newcommand{\f}{\mathcal{F}}

\newcommand{\pow}{\wp}

\newcommand{\N}{\mathbb{N}}

\newcommand{\param}[1]{\mathsf{param}\textup{-}#1}

\newcommand{\dom}{\mathsf{dom}}

\newcommand{\org}{\mathrm{org}^+}
\newcommand{\lay}{\mathrm{lay}^+}

\newcommand{\und}[1]{\underline{#1}}

\newcommand{\clo}{\mathsf{closure}}


\newcommand{\thick}{\mathsf{thick}}
\newcommand{\thickl}{\thick_l}
\newcommand{\localthickl}{\mathsf{local}\textup{-}\thickl}
\newcommand{\quantthickl}{\mathsf{quant}\textup{-}\thickl}

\newcommand{\lowerdeg}{\mathsf{lower}\textup{-}\mathsf{deg}}

\newcommand{\restrict}{\upharpoonright}

\renewcommand{\nu}[1]{\textup{{\small $\mathsf{nu}$}-{$ #1 $}}}

\newcommand{\case}[1]{\textup{{\small $\mathsf{case}$}-{$ #1 $}}}

\newcommand{\coclique}{\textup{{\small  $\mathsf{co}$}-{\small $\mathsf{CLIQUE}$}}}
\newcommand{\clique}{\textup{{\small $\mathsf{CLIQUE}$}}}

\newcommand{\caseclique}{\case{\clique}}
\newcommand{\casecoclique}{\case{\coclique}}

\newcommand{\fpt}{\textup{\small $\mathsf{FPT}$}}
\newcommand{\wone}{\textup{\small $\mathsf{W[1]}$}}
\newcommand{\cowone}{\textup{\small $\mathsf{co}$-$\mathsf{W[1]}$}}

\renewcommand{\S}{\mathcal{S}}
\newcommand{\G}{\mathcal{G}}

\newcommand{\image}{\mathsf{image}}

\newcommand{\str}{\mathrm{STR}}

\begin{abstract}
\begin{quote}
Many natural combinatorial quantities can be expressed by
counting the number of homomorphisms to a fixed relational structure.
For example, the number of 3-colorings of an undirected graph $G$
is equal to the number of homomorphisms from $G$ to the
$3$-clique.  In this setup, the structure receiving the homomorphisms is
often referred to as a \emph{template}; we use the term
\emph{template function} to refer to a function, from
structures to natural numbers,
that is definable as the number
of homomorphisms to a fixed template.
There is a literature that studies the complexity of
template functions.

The present work is concerned with relating template functions
to the problems of counting, with respect to various fixed templates,
 the number of two particular types of homomorphisms:
 surjective homomorphisms and what we term \emph{condensations}.
A surjective homomorphism is a homomorphism that 
maps the universe of the first structure surjectively onto
the universe of the second structure;
a condensation is a homomorphism that, in addition,
maps each relation of the first structure surjectively onto
the corresponding relation of the second structure.

In this article, we explain how any problem
of counting surjective homomorphisms to a fixed template 
is polynomial-time equivalent to computing a linear combination
of template functions; we also show this
for any problem of counting condensations to a fixed template.
Via a theorem that characterizes
the complexity of computing such a linear combination,
we show how
a known dichotomy for template functions
can be used to infer a dichotomy for counting surjective
homomorphisms on fixed templates,
and likewise a dichotomy for counting condensations
on fixed templates.
Our study is strongly inspired by, based on,
and can be viewed as a dual of
the graph motif framework of Curticapean, Dell, and Marx (STOC 2017); 
that framework 
is in turn based on work of Lov\'asz (2012).
\end{quote}
\end{abstract}

\newpage

\section{Preliminaries}

When $f: A \to B$ is a map and $A' \subseteq A$,
we use $f(A')$ to denote the
set $\{ f(a) ~|~ a \in A' \}$.

\subsection{Structures, homomorphisms, and company}

A \emph{signature} is a set of \emph{relation symbols};
each relation symbol $R$ has an associated arity (a natural number),
denoted by $\ar(R)$.
A \emph{structure} $\relb$ over signature $\sigma$
consists of a \emph{universe} $B$ which is a set,
and an interpretation $R^{\relb} \subseteq B^{\ar(R)}$
for each relation symbol $R \in \sigma$.
We use $||\relb||$ to denote the \emph{total size}
of $\relb$, defined as
$|B| + \sum_{R \in \sigma} |R^{\relb}|$.
We will in general use the symbols $\rela, \relb, \ldots$ to denote
structures, and the symbols $A, B, \ldots$ to denote their respective
universes.
In this article, we 
assume that signatures under discussion are finite, and 
assume that all structures under discussion are \emph{finite}; 
a structure is \emph{finite}
if its universe is finite.

Let $\relb$ be a structure over signature $\sigma$.
When $B' \subseteq B$, we define
$\relb[B']$ as the structure
with universe $B'$
and where $R^{\relb[B']} = R^{\relb} \cap B'^{\ar(R)}$.
We define an \emph{induced substructure} of $\relb$ to 
be a structure of the form $\relb[B']$, where $B' \subseteq B$.
Observe that a structure $\rela$ has $2^{|A|}$ induced substructures.
We define a \emph{deduct} of $\relb$ to be a structure
obtained from $\relb$ by removing tuples from relations of $\relb$,
that is, a structure $\relc$ (over signature $\sigma$)
is a \emph{deduct} of $\relb$ 
if $C = B$ and, for each $R \in \sigma$,
it holds that $R^\relc \subseteq R^{\relb}$.

Let $\rela$ and $\relb$ be structures over the same signature $\sigma$.
A \emph{homomorphism} from $\rela$ to $\relb$
is a map $h: A \to B$ such that
for each relation symbol $R \in \sigma$,
it holds that $h(R^{\rela}) \subseteq R^{\relb}$.
A \emph{surjective homomorphism} from $\rela$ to $\relb$
is a homomorphism such that $h(A) = B$, that is, 
such that $h$ is surjective
as a mapping from the set $A$ to the set $B$.
A \emph{condensation} from $\rela$ to $\relb$
is a surjective homomorphism satisfying the condition that
for each relation symbol $R \in \sigma$,
it holds that $h(R^{\rela}) = R^{\relb}$.
This condition is sometimes referred to as \emph{edge-surjectivity}
in graph-theoretic contexts.\footnote{
We remark that some authors use the term \emph{surjective homomorphism}
to refer to what we refer to as a condensation.}
Notions similar to the notion of \emph{condensation}
have been studied in the literature: 
notably, the term \emph{compaction} is sometimes used 
(for example, in~\cite{FockeGoldbergZivny17-counting-surjhoms-and-compactions})
to refer to
a homomorphism between graphs that maps the edge
relation of the first graph surjectively onto
the relation that contains the non-loop edges of the second graph.

Two structures $\relb$, $\relb'$ are \emph{homomorphically equivalent}
if there exists a homomorphism from $\relb$ to $\relb'$
and
there exists a homomorphism from $\relb'$ to $\relb$.

Throughout, we tacitly use the fact that the composition of
a homomorphism from $\rela$ to $\relb$
and a homomorphism from $\relb$ to $\relc$
is a homomorphism from $\rela$ to $\relc$.

\subsection{Computational problems}

\newcommand{\chom}{\ensuremath{\sharp\mathsf{HOM}}}
\newcommand{\cshom}{\ensuremath{\sharp\mathsf{SURJHOM}}}
\newcommand{\ccondens}{\ensuremath{\sharp\mathsf{CONDENS}}}

\renewcommand{\hom}{\ensuremath{\mathrm{Hom}}}
\newcommand{\shom}{\ensuremath{\mathrm{Surjhom}}}
\newcommand{\condens}{\ensuremath{\mathrm{Condens}}}
\newcommand{\indsub}{\ensuremath{\mathrm{Indsub}}}
\newcommand{\deducts}{\ensuremath{\mathrm{Deducts}}}

We now define the computational problems to be studied.
For each structure $\relb$ over signature $\sigma$:

\begin{itemize}

\item 
Define
$\chom(\relb)$ to be the problem of computing,
given a structure $\rela$ over signature $\sigma$,
the number of homomorphisms from $\rela$ to $\relb$.

\item 
Define
$\cshom(\relb)$ to be the problem of computing,
given a structure $\rela$ over signature $\sigma$,
the number of surjective homomorphisms from $\rela$ to $\relb$.

\item 
Define
$\ccondens(\relb)$ to be the problem of computing,
given a structure $\rela$ over signature $\sigma$,
the number of condensations from $\rela$ to $\relb$.

\end{itemize}

\section{Linear combinations of homomorphisms}

Our development is strongly inspired by and based on 
the framework of Curticapean, Dell, and Marx~\cite{CurticapeanDellMarx17-homomorphisms-good-basis},
which in turn was based on work of Lov\'asz~\cite{Lovasz67-operations-structures,Lovasz12-book-networks-and-limits}.
It is also informed by the theory developed by the 
current author with Mengel~\cite{ChenMengel14-pp-arxiv,ChenMengel15-pp-icdt,ChenMengel16-pods-ep-counting,ChenMengel17-lics-logic-counting}.
In these works, a dual setup is considered, where
one fixes the structure $\rela$ from which homomorphisms originate,
and counts the number of homomorphisms that an input structure
receives from $\rela$.
Many of our observations and results can be seen to have duals
in the cited works.

For each signature $\sigma$,
let $\str[\sigma]$ denote the class of all structures over $\sigma$,
and
fix
$\str^*[\sigma]$ to be a subclass of $\str[\sigma]$
that contains exactly one structure from
each isomorphism class of structures contained in $\str[\sigma]$.

For structures $\rela$, $\relb$ over the same signature,
we use:
\begin{itemize}

\item 
$\hom(\rela,\relb)$ to denote the number of homomorphisms
from $\rela$ to $\relb$,

\item
$\shom(\rela,\relb)$ to denote the number of surjective homomorphisms
from $\rela$ to $\relb$,

\item
$\condens(\rela, \relb)$ to denote the number of
condensations from $\rela$ to $\relb$,

\item
$\indsub(\relb', \relb)$ to denote the number of induced substructures
of $\relb$ that are isomorphic to $\relb'$, and

\item
$\deducts(\relb', \relb)$ to denote the number of deducts of
$\relb$ that are isomorphic to $\relb'$.

\end{itemize}
We use $\hom(\cdot,\relb)$ to denote the mapping
that sends a structure $\rela$ to $\hom(\rela,\relb)$,
and use $\shom(\cdot,\relb)$, etc. analogously.

Observe that
\begin{equation}
\label{eq:hom}
\hom(\rela,\relb) = \sum_{\relb' \in \str^*[\sigma]} 
\shom(\rela,\relb') \cdot \indsub(\relb',\relb).
\end{equation}

We briefly justify this as follows.  Each homomorphism
$h$ from $\rela$ to $\relb$ is a surjective homomorphism
from $\rela$ onto an induced substructure of $\relb$, namely,
onto $\relb[h(A)]$.  
Let $\relb' \in \str^*[\sigma]$ be 
isomorphic to an induced substructure of $\relb$, and let us count
the number of homomorphisms $h$ from $\rela$ to $\relb$ such that
$\relb[h(A)]$ is isomorphic to $\relb'$.
Let $\relb_1, \ldots, \relb_k$ be a list of all induced substructures
of $\relb$ that are isomorphic to $\relb'$.
Then, we have 
$\shom(\rela,\relb_1) = \cdots = \shom(\rela,\relb_k) = \shom(\rela,\relb')$
and $k = \indsub(\relb',\relb)$,
so the desired number is 
$\shom(\rela,\relb_1) + \cdots + \shom(\rela,\relb_k)$,
which is equal to
$\shom(\rela,\relb') \cdot \indsub(\relb',\relb)$.

Observe that 
\begin{equation}
\label{eq:shom}
\shom(\rela,\relb) = 
\sum_{\relb' \in \str^*[\sigma]}
\condens(\rela,\relb') \cdot \deducts(\relb', \relb).
\end{equation}
The justification for this equation has the same flavor as
that of the previous equation.  
Each surjective homomorphism $h$ from $\rela$ to $\relb$
is a condensation from $\rela$ to a deduct of $\relb$;
when $\relb'$ is isomorphic to a deduct of $\relb$,
the product $\condens(\rela,\relb') \cdot \deducts(\relb', \relb)$
is the number of condensations from $\rela$ to 
a deduct of $\relb$ that is isomorphic to $\relb'$.

It is direct from Equation~\ref{eq:hom}
that 
\begin{equation}
\label{eq:shom-new}
\shom(\rela,\relb) = \hom(\rela,\relb) - 
\sum_{\relb' \in \str^*[\sigma], |B'| < |B|} 
\shom(\rela,\relb') \cdot \indsub(\relb',\relb).
\end{equation}
From this, one can straightforwardly verify by induction on $|B|$
that the function $\shom(\cdot,\relb)$
can be expressed as a linear combination of functions
each having the form
$\hom(\cdot,\relc)$; moreover, such a linear combination is
computable from $\relb$.  We formalize this as follows.

\begin{prop}
\label{prop:shom}
There exists an algorithm that,
given as input a structure $\relb$ over signature $\sigma$,
outputs a list 
$(\beta_1,\relb_1),\ldots,(\beta_k,\relb_k) \in \Q \times \str[\sigma]$,
where the 
values $\beta_i$ are non-zero
and the structures
$\relb_i$ are pairwise non-isomorphic and
such that, for all structures $\rela$, it holds that
$$\shom(\rela,\relb) = 
\beta_1 \cdot \hom(\rela,\relb_1)
+ \cdots + 
\beta_k \cdot \hom(\rela,\relb_k).$$
\end{prop}

In an analogous fashion,
it is direct from Equation~\ref{eq:shom}
that
\begin{equation}
\label{eq:condens-new}
\condens(\rela,\relb) = \shom(\rela,\relb) -
\sum_{\relb' \in \str^*[\sigma], |B'| = |B|, ||\relb'|| < ||\relb||}
\condens(\rela,\relb') \cdot \deducts(\relb', \relb).
\end{equation}
One can verify by induction that the function $\condens(\cdot,\relb)$
can be expressed as a linear combination of functions each having the form
$\shom(\cdot,\relc)$; such a linear combination is computable from
$\relb$, and so in conjunction with Proposition~\ref{prop:shom},
we obtain the following.

\begin{prop}
\label{prop:condens}
There exists an algorithm that,
given as input a structure $\relb$ over signature $\sigma$,
outputs a list 
$(\beta_1,\relb_1),\ldots,(\beta_k,\relb_k) \in \Q \times \str[\sigma]$,
where the 
values $\beta_i$ are non-zero
and the structures
$\relb_i$ are pairwise non-isomorphic and
such that, for all structures $\rela$, it holds that
$$\condens(\rela,\relb) = 
\beta_1 \cdot \hom(\rela,\relb_1)
+ \cdots + 
\beta_k \cdot \hom(\rela,\relb_k).$$
\end{prop}

\begin{remark}
We can write
Equation~\ref{eq:hom} 
in the following form:
$$\hom(\rela,\relb) = \sum_{\relb'} \shom(\rela,\relb'),$$
where the sum is over all induced substructures $\relb'$ of $\relb$;
analogously, we can write
Equation~\ref{eq:shom} 
in the following form:
$$\shom(\rela,\relb) = \sum_{\relb'} \condens(\rela,\relb'),$$
where the sum is over all deducts $\relb'$ of $\relb$.
From these forms, one can use M\"obius inversion on posets
to express $\shom(\cdot, \relb)$ as a linear combination
of functions $\hom(\cdot, \relb)$;
and likewise to express $\condens(\cdot,\relb)$
as a linear combination of functions $\shom(\cdot,\relc)$,
which linear combination can then be expressed as a linear combination of 
functions $\hom(\cdot,\relb)$.
\end{remark}

\begin{remark}
Equations~\ref{eq:hom} and~\ref{eq:shom}
can be conceived of as matrix identities.
Let $\hom^*$ denote the restriction of $\hom$
to pairs in $\str^*[\sigma] \times \str^*[\sigma]$,
and view it as an infinite matrix whose indices are such pairs
and having entries in $\Q$;
define and view $\shom^*$, etc. analogously.
Then Equation~\ref{eq:hom}, in matrix notation, is expressed by
$$\hom^* = \shom^* \cdot \indsub^*.$$
Analogously, Equation~\ref{eq:shom}, in matrix notation,
is expressed by
$$\shom^* = \condens^* \cdot \deducts^*.$$

Suppose that, for the indexing,
the structures in $\str^*[\sigma]$ 
are ordered in a way that respects total size,
that is, whenever $\relb$ comes before $\relb'$,
it holds that $||\relb|| \leq ||\relb'||$.
Then, the matrices $\indsub^*$ and $\deducts^*$
are readily seen to be upper triangular and 
to have all diagonal entries equal to $1$;
it can be verified that they are invertible.
\end{remark}

\section{The space of template parameters}

We now study the space of linear combinations of functions
$\hom(\cdot, \relb)$.
Fix $\sigma$ to be a signature.
Define a \emph{template function} to be a function
$f: \str[\sigma] \to \Q$ such that
there exists a structure $\relb \in \str[\sigma]$
where, for each $\rela \in \str[\sigma]$,
it holds that $f(\rela) = \hom(\rela,\relb)$.
Define a \emph{template parameter}
to be a function
$f: \str[\sigma] \to \Q$ 
that can be expressed as a finite linear combination of
template functions.
Template parameters naturally form a vector space,
and this space is clearly spanned by the template functions.
We prove that the template functions
$(\hom(\cdot, \relb))_{\relb \in \str^*[\sigma]}$
are linearly independent, and hence form a basis for this vector space.

\begin{theorem}
\label{thm:independence}
Let $(\beta_1,\relb_1),\ldots,(\beta_n,\relb_n) \in \Q \times \str^*[\sigma]$
be such that the $\relb_i$ are pairwise distinct.
Suppose that, for all structures $\rela \in \str[\sigma]$,
it holds that $\sum_{i=1}^n \beta_i \hom(\rela,\relb_i) = 0$.
Then $\beta_1 = \cdots = \beta_n = 0$.
\end{theorem}

We first establish a lemma.

\begin{lemma}
\label{lemma:lovasz-variant}
Suppose that $\relb_1,\ldots,\relb_k \in \str^*[\sigma]$
are pairwise distinct, but all homomorphically equivalent.
Then, there exists a structure $\rela_k$
such that the values $(\hom(\rela_k,\relb_i))_{i=1,\ldots,k}$
are non-zero and pairwise distinct.
\end{lemma}

For two structures $\rela_1$, $\rela_2$, we use
$\rela_1 + \rela_2$ to denote their disjoint union; and,
for $N \geq 0$, we use $N \rela_1$ to denote the $N$-fold
disjoint union of $\rela_1$ with itself.
The identity
$\hom(\rela_1+\rela_2,\relb) = \hom(\rela_1,\relb)\cdot\hom(\rela_2,\relb)$
is known and straightforwardly verified.

\begin{proof}
We prove this by induction.  In the case that $k = 1$,
one can simply take $\rela_1 = \relb_1$.

Suppose that $k > 1$.
By induction, there exists $\rela_{k-1}$
such that 
$(\hom(\rela_{k-1},\relb_i))_{i=1,\ldots,k-1}$
are non-zero and pairwise distinct.
Let us assume for the sake of notation
that $\hom(\rela_{k-1},\relb_1) < \cdots < \hom(\rela_{k-1},\relb_{k-1})$.
Since the structures $\relb_i$ are homomorphically equivalent,
we have $\hom(\rela_{k-1},\relb_k) > 0$.
If $\hom(\rela_{k-1},\relb_k)$ 
is distinct from
each of the values $(\hom(\rela_{k-1},\relb_i))_{i=1,\ldots,k-1}$,
we are done.
Otherwise, there exists a
unique index $\ell \in \{ 1, \ldots, k-1 \}$ such that
$\hom(\rela_{k-1},\relb_k) = \hom(\rela_{k-1},\relb_\ell)$.
By Lovasz's theorem~\cite{Lovasz67-operations-structures}, there exists a structure $\rela'$
such that 
$\hom(\rela',\relb_k) \neq \hom(\rela',\relb_\ell)$;
observe that 
since $\relb_k$ and $\relb_\ell$ are homomorphically equivalent,
both of these values are non-zero;
indeed, all of the values $(\hom(\rela',\relb_i))_{i = 1,\ldots,k}$ 
are non-zero.

We claim that for all sufficiently large values $M$,
the structure $\rela_k = M \rela_{k-1} + \rela'$
has the desired property that
the values $(\hom(\rela_k,\relb_i))_{i=1,\ldots,k}$
are non-zero and pairwise distinct.
This is indeed straightforward to verify.
We have $\hom(M \rela_{k-1} + \rela',\relb_k)
= \hom(\rela_{k-1},\relb_k)^M \cdot \hom(\rela',\relb_k) > 0$,
and since the structures $\relb_i$ are homomorphically equivalent,
we obtain that
the values $(\hom(\rela_k,\relb_i))_{i=1,\ldots,k}$
are non-zero.
Let us now argue pairwise distinctness.
When $j$ is such that $1 \leq j < i$,
for sufficiently large values of $M$,
it will hold that
$\frac{\hom(\rela',\relb_j)}{\hom(\rela',\relb_{j+1})} <
(\frac{\hom(\rela_{k-1},\relb_{j+1})}{\hom(\rela_{k-1},\relb_j)})^M$,
from which it follows that
$ \hom(M \rela_{k-1} + \rela',\relb_j) 
< \hom(M \rela_{k-1} + \rela',\relb_{j+1}) $.
In a similar way, one sees 
that when $j \in \{ 1, \ldots, k-1 \} \setminus \{ \ell \}$,
for sufficiently large values of $M$,
it holds that
$\hom(M \rela_{k-1} + \rela',\relb_j) \neq
\hom(M \rela_{k-1} + \rela',\relb_k)$.
Finally, we have for all $M \geq 1$ that
$\hom(M \rela_{k-1} + \rela',\relb_\ell) \neq
\hom(M \rela_{k-1} + \rela',\relb_k)$,
as a consequence of 
$\hom(\rela_{k-1},\relb_k) = \hom(\rela_{k-1},\relb_\ell)$ and
$\hom(\rela',\relb_k) \neq \hom(\rela',\relb_\ell)$.
\end{proof}

\begin{proof} (Theorem~\ref{thm:independence})
We prove this by induction on $n$.  It is clear for $n = 1$,
so suppose that $n > 1$.

We assume for the sake of notation that $\relb_1$
is extremal in that for each other structure $\relb_j$,
either $\relb_1$ is homomorphically equivalent to $\relb_j$
or does not admit a homomorphism to $\relb_j$.
We assume further that $\relb_1, \ldots, \relb_m$
is a list of the structures among $\relb_1, \ldots, \relb_n$
that are homomorphically equivalent to $\relb_1$.

Applying Lemma~\ref{lemma:lovasz-variant}
to $\relb_1, \ldots, \relb_m$, 
we obtain a structure $\rela$ such that
the values $(\hom(\rela,\relb_i))_{i=1,\ldots,m}$
are pairwise distinct.
Consider the structures $(\rela_{k})_{k=0,\ldots,m-1}$
defined by $\rela_k = k \rela + \relb_1$.
For each $k \in \{ 0, \ldots, m-1 \}$,
we have
$\sum_{i=1}^n \beta_i \hom(\rela_k,\relb_i) = 0$,
which implies 
$\sum_{i=1}^n \beta_i \hom(\relb_1,\relb_i) \hom(\rela,\relb_i)^k = 0$,
which in turn implies 
$\sum_{i=1}^m \beta_i \hom(\relb_1,\relb_i) \hom(\rela,\relb_i)^k = 0$.
Now, form a system of equations by taking this last equation 
over $k \in \{ 0, \ldots, m-1 \}$;
view it as a system of equations over unknowns 
$y_i = \beta_i \hom(\relb_1,\relb_i)$, where $i$ ranges from $1$ to $m$.
The corresponding matrix is a Vandermonde matrix,
implying that $y_1 = \cdots = y_m = 0$.
Since the values $(\hom(\relb_1,\relb_i))_{i=1,\ldots,m}$
are all non-zero, we infer that
$\beta_1 = \cdots = \beta_m = 0$.
By applying induction, we obtain that
$\beta_{m+1} = \cdots = \beta_n = 0$.
\end{proof}

\section{The complexity of template parameters}

We now study the complexity of computing template parameters,
showing in essence that computing a template parameter 
has the same complexity as being able to compute
all of its constituent functions $\hom(\cdot,\relb)$.

\begin{theorem}
\label{thm:turing-equivalent}
Let $(\beta_1,\relb_1),\ldots,(\beta_n,\relb_n) \in \Q \times \str^*[\sigma]$
be such that the values $\beta_i$ are non-zero and such that 
the structures $\relb_i$ are pairwise non-isomorphic.
\begin{itemize}

\item 
Let $f: \str[\sigma] \to \Q$ be the function defined by
$f(\rela) = \sum_{i=1}^n \beta_i \cdot \hom(\rela,\relb_i)$.

\item
Let $g: \{ 1, \ldots, n \} \times \str[\sigma] \to \Q$
be the function defined by
$g(i,\rela) = \hom(\rela,\relb_i)$.

\end{itemize}
The functions $f$ and $g$ are equivalent under polynomial-time
Turing reduction.
\end{theorem}

\newcommand{\tred}{\leq_T^p}
For functions $h, h'$,
we use $h \tred h'$ to indicate that $h$ polynomial-time Turing 
reduces to $h'$.

\begin{proof}
It is clear that $f \tred g$, so we prove that $g \tred f$,
by induction on $n$; the result is clear for $n=1$.
By rearranging indices if necessary, let us assume that
the structures $\relb_1, \ldots, \relb_m$ are as described
in the second paragraph of the proof of Theorem~\ref{thm:independence}.
Let $g_1$ be the restriction of $g$ to
$ \{ 1, \ldots, m \} \times \str[\sigma]$,
and
let $g_2$ be the restriction of $g$ to
$ \{ m+1, \ldots, n \} \times \str[\sigma]$.
Let $f_2: \str[\sigma] \to \Q$ be the function defined by
$f(\rela) = \sum_{i=m+1}^{n} \beta_i \cdot \hom(\rela,\relb_i)$.

Let us show $g_1 \tred f$.
By applying Lemma~\ref{lemma:lovasz-variant}
to $\relb_1, \ldots, \relb_m$,
we obtain a structure $\rela'$ such that
the values $(\hom(\rela',\relb_i))_{i=1,\ldots,m}$
are pairwise distinct.
Given a pair $(j,\rela)$ as input,
the reduction constructs
the structures $(\rela_k)_{k=0,\ldots,m-1}$
defined by $\rela_k = \relb_1 + \rela + k \rela'$,
and then computes the various values $f(\rela_k)$.
We have, for each $k$,
$\sum_{i=1}^n \beta_i \hom(\rela_k,\relb_i) = f(\rela_k)$;
from this, we obtain
$\sum_{i=1}^m \beta_i \hom(\relb_1,\relb_i) \hom(\rela,\relb_i) \hom(\rela',\relb_i)^k = f(\rela_k)$.
Viewing this as a system of equations over unknowns
$y_i = \beta_i \hom(\relb_1,\relb_i) \hom(\rela,\relb_i)$,
the corresponding matrix is Vandermonde.
Hence, we may solve for these unknowns $y_i$,
and then from their solution compute the values $\hom(\rela,\relb_i)$.
We then output the desired value $\hom(\rela,\relb_j)$.

We now argue that $f_2 \tred f$.
Given a structure $\rela$ as input,
the reduction first computes $f(\rela)$.  
Since we just showed that $g_1 \tred f$,
the reduction may also compute the values 
$\hom(\rela,\relb_1),\ldots,\hom(\rela,\relb_m)$.
By subtracting 
$\beta_1 \hom(\rela,\relb_1) + \cdots + \beta_m \hom(\rela,\relb_m)$
from $f(\rela)$, the desired value $f_2(\rela)$ is computed.

We obtain $g_2 \tred f_2$ by induction; it follows that $g_2 \tred f$.

As we established that $g_1 \tred f$ and $g_2 \tred f$,
it is immediate that $g \tred f$.
\end{proof}

\section{Complexity results}
\newcommand{\fp}{\mathsf{FP}}
\newcommand{\sharpp}{\mathsf{\sharp P}}

Previous work established a complexity dichotomy 
for the family of problems $\chom(\relb)$.
Let $\fp$ denote the functional version of polynomial time.
A criterion was presented that distinguishes the structures $\relb$
for which $\chom(\relb)$ is in the class $\fp$,
from those that are complete for $\sharpp$.  
Here, we refer to this criterion
as the \emph{$\chom(\cdot)$-tractability condition};
we refer the reader to~\cite{DyerRicherby13-effective-dichotomy-counting} for a precise formulation of this criterion.
The dichotomy can be made precise as follows.

\begin{theorem} \cite{Bulatov13-counting,DyerRicherby13-effective-dichotomy-counting}
\label{thm:chom-dichotomy}
Let $\relb$ be any structure.
If $\relb$ satisfies the $\chom(\cdot)$-tractability condition,
then the problem $\chom(\relb)$ is in $\fp$;
otherwise, it is $\sharpp$-complete under polynomial-time Turing reducibility.
\end{theorem}

The following was also established.

\begin{theorem} \cite{DyerRicherby13-effective-dichotomy-counting}
\label{thm:chom-meta-problem}
The $\chom(\cdot)$-tractability condition is decidable.
\end{theorem}

Define the \emph{$\cshom(\cdot)$-tractability condition} to
be satisfied by a structure $\relb$ iff the algorithm of
Proposition~\ref{prop:shom} returns a list
$(\beta_1,\relb_1),\ldots,(\beta_k,\relb_k)$
such that each structure $\relb_i$ satisfies the
$\chom(\cdot)$-tractability condition.
(We remark here that all algorithms behaving as described
in Proposition~\ref{prop:shom} will output the same list,
up to permutation, due to Theorem~\ref{thm:independence}.)
We obtain the following.

\begin{theorem}
\label{thm:cshom-dichotomy}
Let $\relb$ be any structure.
If $\relb$ satisfies the $\cshom(\cdot)$-tractability condition,
then the problem $\cshom(\relb)$ is in $\fp$;
otherwise, it is $\sharpp$-complete under polynomial-time Turing 
reducibility.
Moreover, the $\cshom(\cdot)$-tractability condition is decidable.
\end{theorem}

\begin{proof}
Let
$(\beta_1,\relb_1),\ldots,(\beta_k,\relb_k)$.
be the list obtained by invoking the algorithm of
Proposition~\ref{prop:shom} on $\relb$.

Suppose $\relb$ satisfies the $\cshom(\cdot)$-tractability condition.
Let us argue that $\cshom(\relb)$ is in $\fp$.
The algorithm is given a structure $\rela$ as input.
By assumption, each $\relb_i$ satisfies the 
$\chom(\cdot)$-tractability condition,
and so each of the values $\hom(\rela,\relb_i)$ can be computed
in polynomial time.
The algorithm outputs the sum 
$\beta_1 \cdot \hom(\rela,\relb_1)
+ \cdots + 
\beta_k \cdot \hom(\rela,\relb_k).$

Suppose that $\relb$ does not satisfy 
the $\cshom(\cdot)$-tractability condition.
There exists an index $\ell$ such that $\relb_\ell$
does not satisfy 
the $\chom(\cdot)$-tractability condition,
so $\chom(\relb_\ell)$ is $\sharpp$-complete by
Theorem~\ref{thm:chom-dichotomy}.
Let $f$ and $g$ be the functions described in the statement of
Theorem~\ref{thm:turing-equivalent}.
Clearly,  $\chom(\relb_\ell) \tred g$.
Since $g \tred f$ by
Theorem~\ref{thm:turing-equivalent},
we obtain that $f$ is $\sharpp$-complete, as desired.

Decidability of
the $\cshom(\cdot)$-tractability condition 
is immediate from its definition and Theorem~\ref{thm:chom-meta-problem}.
\end{proof}

Define the \emph{$\ccondens(\cdot)$-tractability condition} to
be satisfied by a structure $\relb$ iff the algorithm of
Proposition~\ref{prop:condens} returns a list
$(\beta_1,\relb_1),\ldots,(\beta_k,\relb_k)$
such that each structure $\relb_i$ satisfies the
$\chom(\cdot)$-tractability condition.
We have the following; the proof is analogous to that of
Theorem~\ref{thm:cshom-dichotomy}.

\begin{theorem}
Let $\relb$ be any structure.
If $\relb$ satisfies the $\ccondens(\cdot)$-tractability condition,
then the problem $\ccondens(\relb)$ is in $\fp$;
otherwise, it is $\sharpp$-complete under polynomial-time Turing 
reducibility.
Moreover, the $\ccondens(\cdot)$-tractability condition is decidable.
\end{theorem}

We would like to present further consequences of our theory.
From Equation~\ref{eq:shom-new}, it can be elementarily verified
that, for any structure $\relb$, the expression 
of $\shom(\cdot,\relb)$ as a linear combination of
functions $\hom(\cdot,\relb')$ gives a coefficient of $1$
to $\hom(\cdot,\relb)$.
The same fact holds for $\condens(\cdot,\relb)$ in place of
$\shom(\cdot,\relb)$, as can be elementarily seen
from Equations~\ref{eq:condens-new} and~\ref{eq:shom-new}.
(That $\hom(\cdot,\relb)$ receives a coefficient of $1$
is these expressions also immediate from M\"obius inversion.)
We thus obtain the following, via Theorem~\ref{thm:turing-equivalent}.

\begin{corollary}
For each structure $\relb$,
the problem $\chom(\relb)$ reduces to $\cshom(\relb)$.
\end{corollary}

\begin{corollary}
For each structure $\relb$,
the problem $\chom(\relb)$ reduces to $\ccondens(\relb)$.
\end{corollary}

In the setting of graphs, results similar to these two corollaries
were obtained by Focke, Goldberg, and Zivny~\cite{FockeGoldbergZivny17-counting-surjhoms-and-compactions}.\footnote{
See their Theorem 30 and Theorem 13.
We remark that their Theorem 13 concerns compactions,
and in their setup, inputs are irreflexive graphs.}
We would like to emphasize that here,
these two corollaries fall out as very simple consequences
of a more general theory.

This work~\cite{FockeGoldbergZivny17-counting-surjhoms-and-compactions}
presented classifications of undirected graph templates
with respect to the problems of counting surjective homomorphisms
and of counting compactions.

Let us mention that,
for the \emph{decision} problem of checking
existence of a surjective homomorphism,
a complexity classification of templates seems to be currently elusive,
although there is work in this direction
(see for example~\cite{Chen14-hardness-surjective-csp,LaroseMartinPaulusma2017-surjectivehom-reflexive-digraphs} and the references therein).

\paragraph{Acknowledgements.} 
The author is grateful to Radu Curticapean and Holger Dell
for 
discussions about
and clear explanations of their joint work~\cite{CurticapeanDellMarx17-homomorphisms-good-basis} with
D\'aniel Marx.
The author thanks Stefan Mengel for his collaboration
on database queries~\cite{ChenMengel14-pp-arxiv,ChenMengel15-pp-icdt,ChenMengel16-pods-ep-counting,ChenMengel17-lics-logic-counting}, in which one can see
effects similar to those in the present work.
This work was supported by the
Spanish Project MINECO COMMAS TIN2013-46181-C2-R, Basque Project GIU15/30, and Basque Grant UFI11/45.

\bibliographystyle{plain}

\bibliography{/Users/hubiec/Dropbox/active/writing/hubiebib.bib}

\end{document}